\documentclass[a4paper,12pt]{article}
\usepackage{amsmath}
\usepackage{amsfonts}
\usepackage{amssymb}
\usepackage{amsthm}
\usepackage{graphicx}
\usepackage{ifthen}
\usepackage{verbatim}
\usepackage{amsmath}
\usepackage{setspace}
\usepackage{amsfonts}
\usepackage{amssymb}
\usepackage{amsthm}
\usepackage{lscape}
\usepackage{setspace}
\usepackage{listings}
\usepackage{amsmath}
\usepackage{amsfonts}
\usepackage{amssymb}
\usepackage{amsthm}
\usepackage{ifthen}
\usepackage{verbatim}
\newtheorem{thm}{Theorem}[section]
\newtheorem{ex}[thm]{Example}
\newtheorem{coro}[thm]{Corollary}

\newtheorem{rem}[thm]{Remark}
\newtheorem{prop}[thm]{Proposition}
\newtheorem{proposal}[thm]{Proposal}

\newcommand{\Poincare}{{Poincar\'e }}
\newcommand{\Mobius}{{M\"obius }}

\def\1{\mathbf{1}}
\def\R{\mathbb{R}}
\def\P{\mathbb{P}}

\def\Q{\mathbb{Q}}
\def\E{\mathbb{E}}

\def\dd{\mathrm{d}}

\def\tp{\top}
\def\F{\mathcal{F}}

\begin{document}
\title{Isometric Invariant Quantification of  \\
Gaussian Divergence over \Poincare Disc}

\author{Levent Ali Meng\"ut\"urk\thanks{University College London, Department of Mathematics, l.menguturk@ucl.ac.uk, \, Artificial Intelligence and Mathematics Research Lab, levent@aimresearchlab.com}}
\date{}
\maketitle

\begin{abstract}
The paper presents a statistical geometric duality between the spherical squared-Hellinger distance and a hyperbolic isometric invariant of the \Poincare disc under the action of the general \Mobius group. Motivated by the geometric connection, we propose the usage of the $\mathcal{L}^{2}$-embedded hyperbolic isometric invariant as an alternative way to quantify statistical divergence between Gaussian measures as a contribution to information theory and machine learning.
\end{abstract} 

\section{Introduction}
In recent years, statistical machine learning has increasingly encountered high-dimensional, hierarchically structured data that standard Euclidean metrics fail to capture adequately. Canonically, when comparing two Gaussian distributions, traditional divergence measures, such as the Kullback-Leibler divergence, yield metric distortions by treating the underlying parameter space as flat, forcing Euclidean assumptions onto a geometric structure that is inherently curved. Essentially, since the space of Gaussian distributions naturally inherits a hyperbolic geometry of constant negative curvature via the Fisher information metric, using flat distances that deviate from the induced geodesic path constitutes a geometric distortion in the representation of data. For example, for tree-like hierarchical relationships, using an Euclidean metric flattens the tree, which implicitly forces straight-line paths, where curved-paths are natural. In order to harness the geometric advantages, this paper proposes a new statistical divergence metric rooted in hyperbolic geometry of the \Poincare disc, providing a mathematically principled and robust framework for measuring distributional differences in non-Euclidean machine learning tasks.

The main objective of this work is twofold: (i) present a geometric relationship between the spherical squared-Hellinger distance -- commonly used in machine learning and artificial intelligence -- and a hyperbolic isometric invariant of the \Poincare disc under the action of the general \Mobius group, and (ii) provide a closed-form equation for the proposed $\mathcal{L}^{2}$-embedded statistical divergence measure. More specifically, motivated by the spherical-hyperbolic geometric duality, we propose an $\mathcal{L}^{2}$-embedded hyperbolic isometric invariant as an alternative way to quantify statistical divergence between two Gaussian probability measures. We shall keep this note reasonably concise as a theoretical contribution to statistical information theory and machine learning, and leave its applications using real data for future. To the best of our knowledge, the most relevant observation was drafted in \cite{9}, which we shall formalise, deepen and generalise hereafter.

For the rest of this paper, consider a measurable space $(\Omega,\F)$ endowed with probability measures $\P$ and $\Q$ that are absolutely continuous with respect to the Lebesgue measure $\mathbb{L}$. To minimize notational burden, we shall omit explicit reference to $\mathbb{L}$ and local coordinates, using the shorthand $\dd\mathbb{P}$ (and $\dd\mathbb{Q}$) directly to represent the Radon-Nikodým density $\frac{\dd\mathbb{P}}{\dd\mathbb{L}}$ (and $\frac{\dd\mathbb{Q}}{\dd\mathbb{L}}$) under the integral sign. For our purposes, $\P$ and $\Q$ are Gaussian measures. Using the language of measure theory, the squared-Hellinger distance between $\P$ and $\Q$ (see \cite{aa1}), which we denote as $\Phi(\P,\Q)$, is given by
\begin{align}
\Phi(\P,\Q) = \frac{1}{2}\int_{\Omega}\left( \sqrt{\dd\P(\omega)} - \sqrt{\dd\Q(\omega)} \right)^2. \label{hellinger}
\end{align}
We are now in position to present our main statement before providing the details in the next section.
\begin{proposal}
\label{dualitystatementinitalpresentation}
Define $\Psi(\P,\Q)$ as follows:
\begin{align} 
\Psi(\P,\Q) = \frac{2\int_{\Omega}\left( \dd\P(\omega) - \dd\Q(\omega) \right)^2}{\left(1 - \int_{\Omega}\dd\P(\omega)^2\right)\left(1 - \int_{\Omega}\dd\Q(\omega)^2\right)}. \label{measuretheorticexpinitialpresentation}
\end{align}
Then, $\Phi$ in (\ref{hellinger}) and $\Psi$ in (\ref{measuretheorticexpinitialpresentation}) establish a spherical-hyperbolic duality under an $\mathcal{L}^{2}$-embedding.
\end{proposal}
In the next section, we shall unpack Proposal \ref{dualitystatementinitalpresentation} in detail and provide the mathematical arguments leading to it, which will clarify the statement for the reader. We shall also provide a closed-form solution to (\ref{measuretheorticexpinitialpresentation}) in order to enable its implementation more broadly in machine learning and artificial intelligence applications. The closed-form equation will also allow us to quantify a distributional distance between two Gaussian processes as they progress in time. Finally, we shall provide the parametric asymptotics of (\ref{measuretheorticexpinitialpresentation}) across relevant limits and provide a multivariate generalization.

\section{Geometric Duality}
For the rest of this paper, let $\P$ and $\Q$ be measures on $\R$, where $p:\R\rightarrow\R_+$ and $q:\R\rightarrow\R_+$ are the corresponding Gaussian probability density functions. The choice of the state-space $\R$ here is for parsimony and will later be generalized to $\R^n$ for $n\geq 1$. We identify the parameters of $p$ and $q$ as $\Theta_p = \{\mu_p,\sigma_p^{2}\}$ and $\Theta_q = \{\mu_q,\sigma_q^{2}\}$, respectively, where $\mu_p, \mu_q \in (-\infty,\infty)$ and $\sigma_p,\sigma_q\in(0,\infty)$.
Denoting $||.||_{\mathcal{L}^{2}}$ as the $\mathcal{L}^{2}$-norm and using (\ref{hellinger}), we have the following:
\begin{align}
\Phi(\P,\Q)  &= \frac{1}{2}||\sqrt{p}-\sqrt{q}||^{2}_{\mathcal{L}^{2}}. \label{hellingerl2norm}
\end{align}
Since $p$ and $q$ are non-negative functions and $||\sqrt{p}||_{\mathcal{L}^{2}} = ||\sqrt{q}||_{\mathcal{L}^{2}} = 1$, both $\sqrt{p}$ and $\sqrt{q}$ determine points on the positive-orthant of the unit-sphere $\mathcal{S}^+ \subset \mathcal{L}^{2}$.
\begin{rem}
Note that $(\mathcal{S}, \left\langle . \right\rangle)$ is a Riemannian manifold with constant curvature $K =+1$, given that the inner product $\left\langle . \right\rangle$ on $\mathcal{L}^{2}$ is the Riemannian metric on $\mathcal{S}$. 
\end{rem}
The geodesics between any two points on $\mathcal{S}$ are the great-circles, and accordingly, an angle using the $\mathcal{L}^{2}$-inner product can be defined via
\begin{align} 
\cos \left(\textbf{d}_{\mathcal{S}}(\P,\Q)\right) &= \frac{\left\langle \sqrt{p},\sqrt{q} \right\rangle}{\left\langle \sqrt{p},\sqrt{p} \right\rangle^{\frac{1}{2}} \left\langle \sqrt{q},\sqrt{q} \right\rangle^{\frac{1}{2}}}  =\int_{\R}\sqrt{p(x)}\sqrt{q(x)}\dd x =1-\frac{1}{2}\int_{\R}\left( \sqrt{p(x)} - \sqrt{q(x)} \right)^2\dd x, \notag 
\end{align}
where $\textbf{d}_{\mathcal{S}}(\P,\Q)$ is the Bhattacharyya angle between $\P$ and $\Q$ -- the angle from the center of $\mathcal{S}$ subtended to the endpoints on $\mathcal{S}^{+}$ -- which is equivalent to the spherical distance on $\mathcal{S}^+$ with values in $[0,\pi/2]$. Therefore, using (\ref{hellingerl2norm}), the map $\Phi$ can be rewritten in terms of the cosine of the spherical distance on $\mathcal{S}^+$ as follows:
\begin{align}
\Phi(\P,\Q)  &= 1 - \cos \left(\textbf{d}_{\mathcal{S}}(\P,\Q)\right). \label{bhat}
\end{align}
In information geometry, it is well-known that the parameter space of Gaussian measures forms a 2-dimensional differentiable manifold locally diffeomorphic to $\mathbb{R}^{2}$ (see, for example, \cite{2}), which is a hyperbolic space with constant curvature $K=-1/2$. If this manifold, which we shall denote as $\mathcal{M}$, is endowed with the Fisher information metric, then $\mathcal{M}$ becomes a Riemannian manifold. Accordingly, it is possible to define a distance between $\P$ and $\Q$ via this metric by integrating the infinitesimal line element along the geodesic connecting the corresponding points on $\mathcal{M}$ -- determined by the parameters of $p$ and $q$, respectively -- given by
\begin{align}
\textbf{d}_{\mathcal{M}}(\P,\Q)&=\sqrt{2}\left|\log\left(\frac{1+\zeta(p,q)}{1-\zeta(p,q)}\right)\right| \notag \\
&=2\sqrt{2}\tanh^{-1}(\zeta(p,q)), \label{eq:riemanniandistance1}
\end{align}
where $\zeta$ is a function of the parameters of $p$ and $q$:
\begin{align} \label{eq:FRzeta}
\zeta(p,q)=\left(\frac{(\mu_{q}-\mu_{p})^{2}+2(\sigma_{q}-\sigma_{p})^{2}}{(\mu_{q}-\mu_{p})^{2}+2(\sigma_{q}+\sigma_{p})^{2}} \right)^{\frac{1}{2}}.
\end{align}
The expression in (\ref{eq:riemanniandistance1}) holds when $\mu_{q}\neq\mu_{p}$ and $\sigma_{q}\neq\sigma_{p}$, and the metric takes alternative forms when $\mu_{q}=\mu_{p}$ or $\sigma_{q}=\sigma_{p}$ (see \cite{a1}), which we omit from this note. One can adjust the lengths on $\mathcal{M}$ measured in an alternative unit $R=-\sqrt{-K}/K=\sqrt{2}$, which is analogous to the radian on $\mathcal{S}$.

This geometrical perspective on Gaussian measures is what we shall use to establish a connection between the squared-Hellinger distance and a hyperbolic isometric invariant of the \Poincare disc under the action of the general \Mobius group. First, recall that the general \Mobius group forms a group of transformations of the Riemannian sphere $\mathbb{\overline{C}} = \mathbb{C} \cup \{\infty\}$ (a topological construction, namely one-point compactification), where geometric quantities (e.g. hyperbolic lengths and hyperbolic angles) are invariant under its action. More specifically, a Möbius transformation is a holomorphic function $\eta^{*}:\mathbb{\overline{C}}\rightarrow\mathbb{\overline{C}}$ that satisfies the functional form: 
\begin{align}
\eta^{*}(z)=(az+b)/(cz+d), \notag 
\end{align}
where $a,b,c,d\in\mathbb{C}$ and $ad-bc\ne0$. The general Möbius group, which we shall denote as Möb, is generated by the set of Möbius transformations and the set of complex conjugations, such that $\eta\in \text{Möb}$ is the composition given by 
\begin{align}
\eta=C \circ \eta^{*}_{k} \circ \ldots \circ C \circ \eta^{*}_{1}, \notag 
\end{align}
for some $k\geq1$, where each $\eta^{*}_{j}$ is a Möbius transformation, $C(z)=\overline{z}$ for $z\in\mathbb{C}$ and $C(\infty)=\infty$. We highlight that $C$ is a homeomorphism of $\mathbb{\overline{C}}$, and Möb is equal to the set of homeomorphisms of $\mathbb{\overline{C}}$ that take circles in $\mathbb{\overline{C}}$ to circles in $\mathbb{\overline{C}}$. In fact, it can be shown that elements of Möb are conformal homeomorphisms of $\mathbb{\overline{C}}$, and 
\begin{align}
\text{Möb}(\mathcal{W})=\{\eta\in\text{Möb} \,|\, \eta(\mathcal{W})=\mathcal{W}\}, \notag 
\end{align}
is equal to the group of isometries of the hyperbolic space 
\begin{align}
\label{hyoerbolicspaceone}
\mathcal{W}=\{z\in\mathbb{C} : \Im(z)>0\}, 
\end{align}
given that $\Im(z)$ is the imaginary part of $z$. 
\begin{rem}
The space $\mathcal{W}$ has constant curvature $K=-1$. 
\end{rem}
Using Poincaré uniformization theorem, the metric on $\mathcal{M}$ can be transformed to the metric on $\mathcal{W}$, since $\mathcal{M}$ and $\mathcal{W}$ are conformally equivalent. This can be achieved by adjusting the metric on $\mathcal{M}$ via multiplying it with the positive constant $1/R = 1/\sqrt{2}$. Then, using (\ref{eq:riemanniandistance1}), the distance between the points determined by $p$ and $q$ mapped on $\mathcal{W}$ is given by 
\begin{align}
\textbf{d}_{\mathcal{W}}(\P,\Q)=2\tanh^{-1}(\zeta(p,q)), \label{distanceexpinW} 
\end{align}
where the function $\zeta$ is defined as in (\ref{eq:FRzeta}).
On the other hand, the hyperbolic space of Poincaré disc model is the unit disc $\mathcal{D}$ in the complex plane $\mathbb{C}$ such that
\begin{align}
\label{hyoerbolicspacetwo}
\mathcal{D}=\{z^{*}\in\mathbb{C} : ||z^{*}|| < 1 \},
\end{align}
where $||.||$ is the Euclidean norm. Since $\mathcal{W}$ and $\mathcal{D}$ are both in $\mathbb{\overline{C}}$, it is possible to identify a family of $\eta\in\text{Möb}$, such that $\eta:\mathcal{W}\rightarrow\mathcal{D}$. Essentially, Möb allows to use $\mathcal{W}$ in (\ref{hyoerbolicspaceone}) and $\mathcal{D}$ in (\ref{hyoerbolicspacetwo}) interchangeably when modelling hyperbolic spaces. In particular, if $z$ and $\kappa$ are points in $\mathcal{W}$ and $z^{*}$ is a point in disc $\mathcal{D}$, then
\begin{align} \label{MobspecificHtoD}
z^{*}=e^{i\theta}\frac{z-\kappa}{z-\overline{\kappa}},
\end{align}
is a Möbius transformation conformally mapping $\mathcal{W}$ to $\mathcal{D}$, where $\kappa\in\mathcal{W}$ is an arbitrary point mapped to the center of the disk $\mathcal{D}$, $\theta$ rotates the disk and $z^{*}\in\mathcal{D}$ is the corresponding point of $z\in\mathcal{W}$. 
Accordingly, using Möb, it is possible to compute distances between two points on $\mathcal{D}$ starting from distances on $\mathcal{W}$, or conformally map the points on $\mathcal{W}$ to points on $\mathcal{D}$ using (\ref{MobspecificHtoD}), and calculate the distances on $\mathcal{D}$ directly -- hence, Möb allows for 
\begin{align}
\textbf{d}_{\mathcal{W}}(\P,\Q) \mapsto \textbf{d}_{\mathcal{D}}(\P,\Q). \label{isometrydispres}
\end{align}
Therefore, there is a natural way to characterise Gaussian measures and their distances on the Poincaré disc $\mathcal{D}$ starting from $\mathcal{M}$ followed by $\mathcal{W}$, using Möb.
We are now in the position to highlight an isometric invariant of $\mathcal{D}$ under the action of 
\begin{align}
\text{Möb$(\mathcal{D})=\{\eta\in\text{Möb} \,|\, \eta(\mathcal{D})=\mathcal{D}\}$}, \notag
\end{align}
which we denote by $\Psi$, as follows:
\begin{align} 
\label{eq:isometricinvariantofpoincaredisc11}
\Psi(x^{*},y^{*})=\frac{2||x^{*}-y^{*}||^{2}}{(1-||x^{*}||^{2})(1-||y^{*}||^{2})},
\end{align}
for $||x^{*}|| < 1$ and $||y^{*}||< 1$, where $x^{*}$ and $y^{*}$ are points on $\mathcal{D}$. It can be shown that the hyperbolic isometric invariant $\Psi$ characterises the distance $\textbf{d}_{\mathcal{D}}$ on $\mathcal{D}$ (see \cite{1}). Thus, if we choose $x^{*}$ and $y^{*}$ in (\ref{eq:isometricinvariantofpoincaredisc11}) as points determined by $\P$ and $\Q$ conformally mapped to $\mathcal{D}$ from $\mathcal{W}$, we can write the following representation:
\begin{align} \label{eq:comparisonpoint1}
\Psi(\P,\Q)=\cosh\left(\textbf{d}_{\mathcal{D}}(\P,\Q)\right)-1. 
\end{align}
Using (\ref{distanceexpinW}), (\ref{isometrydispres}) and (\ref{eq:comparisonpoint1}), we have the following:
\begin{align}
\Psi(\mathbb{P}, \mathbb{Q}) &=\cosh\left(\textbf{d}_{\mathcal{D}}(\P,\Q)\right)-1 = \cosh\left(\textbf{d}_{\mathcal{W}}(\mathbb{P},\mathbb{Q})\right) - 1 \notag \\
&= \cosh\left(2\tanh^{-1}(\zeta(\mathbb{P},\mathbb{Q}))\right) - 1 = \frac{2\zeta^2(\mathbb{P},\mathbb{Q})}{1 - \zeta^2(\mathbb{P},\mathbb{Q})}. \label{closedformnewmetricexp}
\end{align}
Thus, plugging (\ref{eq:FRzeta}) into (\ref{closedformnewmetricexp}), we get
\begin{align}
\Psi(\mathbb{P}, \mathbb{Q}) = \frac{(\mu_p - \mu_q)^2 + 2(\sigma_p - \sigma_q)^2}{4\sigma_p\sigma_q}. \label{closefformfinalpoincare}
\end{align}
Note that the closed-form in (\ref{closefformfinalpoincare}), when compared to (\ref{eq:FRzeta}) cancels out the mean parameters from the denominator, while retaining the Fisher information weighting on the variance component in the numerator. In addition, (\ref{closefformfinalpoincare}) provides a simpler mathematical form when compared to (\ref{eq:riemanniandistance1}).

It can be observed that the spherical squared-Hellinger distance $\Phi$ in (\ref{bhat}) is closely related to the hyperbolic isometric invariant $\Psi$ in (\ref{eq:comparisonpoint1}) from a geometric stance. We shall further highlight the following characteristics:
\begin{enumerate}
\item The curvature of $\mathcal{S}$ and that of $\mathcal{D}$ are opposite in sign with $K = +1$ for $\mathcal{S}$ and $K = -1$ for $\mathcal{D}$, which is reflected in the functional forms in (\ref{bhat}) and (\ref{eq:comparisonpoint1}), respectively.
\item The spherical cosine function on the spherical space $\mathcal{S}$ is replaced by the hyperbolic cosine function on the hyperbolic space $\mathcal{D}$. 
\item Since $\textbf{d}_{\mathcal{D}}$ is a metric, we must have $\textbf{d}_{\mathcal{D}}(\P,\Q)\geq 0$, which means $\Psi(\P,\Q)\geq 0$ must hold having $\cosh(0)=1$ and $\cosh(r)$ monotonically increasing in $r\in\mathbb{R}_{+}$ -- in particular, $\Psi=0$ when $\P=\Q$, and $\Psi$ is strictly positive otherwise.
\item $\Psi$ is symmetric in $\P$ and $\Q$ such that $\Psi(\P,\Q) = \Psi(\Q,\P)$ -- note that $\Phi$ is also symmetric in $\P$ and $\Q$ such that $\Phi(\P,\Q) = \Phi(\Q,\P)$.
\end{enumerate}
Accordingly, $\Psi$ in (\ref{eq:comparisonpoint1}) hosts desirable properties as a divergence metric, also shared by the spherical squared-Hellinger distance $\Phi$. Remark \ref{dualitystatementone} below serves to highlight the geometric relation, and is what motivates the rest of this paper.
\begin{rem}
\label{dualitystatementone}
In the Gaussian setting, $\Phi$ in (\ref{bhat}) and $\Psi$ in (\ref{eq:comparisonpoint1}) form a spherical-hyperbolic dual in the aforementioned geometric sense.
\end{rem}
Encouraged by Remark \ref{dualitystatementone}, we move further. It is possible to construct a natural embedding of $\mathcal{D}$ into a Hilbert space, e.g. into $\mathcal{L}^{2}$ -- here, natural embedding refers to an embedding where an isometry on $\mathcal{D}$ can be associated to an isometry on a given Hilbert space. Such an embedding will later allow us to go beyond the geometric constraints of a two-dimensional domain and accommodate multivariate Gaussians. Accordingly, extending (\ref{eq:isometricinvariantofpoincaredisc11}) in the $\mathcal{L}^{2}$-sense, and using the language of measure theory as in (\ref{hellinger}), we propose an $\mathcal{L}^{2}$-embedded $\Psi$ (while keeping notations the same) as follows: 
\begin{align} 
\Psi(\P,\Q)&=\frac{2||p-q ||_{\mathcal{L}^{2}}^{2}}{(1-||p||_{\mathcal{L}^{2}}^{2})(1-||q||_{\mathcal{L}^{2}}^{2})} =\frac{2\int_{\Omega}\left( \dd\P(\omega) - \dd\Q(\omega) \right)^2}{\left(1 - \int_{\Omega}\dd\P(\omega)^2\right)\left(1 - \int_{\Omega}\dd\Q(\omega)^2\right)}, \label{measuretheorticexp}
\end{align}
for $||p||_{\mathcal{L}^{2}} < 1$ and $||q||_{\mathcal{L}^{2}}< 1$. Although we used the same notations, $\Psi$ in (\ref{eq:comparisonpoint1}) is not the same as $\Psi$ in (\ref{measuretheorticexp}). However, we shall highlight that (\ref{measuretheorticexp}) is an isometric invariant over the infinite-dimensional real Hilbert ball $\mathbb{B}(\mathcal{L}^2) = \{ f \in \mathcal{L}^2 : \|f\|_{\mathcal{L}^2} < 1 \}$ equipped with its intrinsic Poincar\'e-type hyperbolic metric -- see \cite{10,11} for infinite-dimensional hyperbolic geometry and the theory of intrinsic metrics on unit balls. Hence, as a result of our embedding, $\Psi(\mathbb{P}, \mathbb{Q})$ in (\ref{measuretheorticexp}) inherits an analogous hyperbolic framework that extends beyond the finite-dimensional setup. 
\begin{rem}
\label{dualitystatement}
In the Gaussian setting, $\Phi$ in (\ref{hellinger}) and $\Psi$ in (\ref{measuretheorticexp}) establish a spherical-hyperbolic duality under an $\mathcal{L}^{2}$-embedding, when
\begin{align}
\label{connectionformalised}
\text{$||p||_{\mathcal{L}^{2}} < 1$ \hspace{0.1in} \emph{and} \hspace{0.1in} $||q||_{\mathcal{L}^{2}}< 1$}.
\end{align}
\end{rem}
We are now in position to present the closed-form solution for the hyperbolic dual of the spherical squared-Hellinger distance we propose in (\ref{measuretheorticexp}).
\begin{prop}
\label{mainpropgeometric}
The divergence (\ref{measuretheorticexp}) on $\R$ is given by
\begin{align} 
\label{closedform}
\Psi(\P,\Q)&= \frac{ (\sigma_p\sqrt{\pi})^{-1} - 2\sqrt{2}\lambda(\Theta_p,\Theta_q) + (\sigma_q\sqrt{\pi})^{-1}}{(1 - (2\sigma_p\sqrt{\pi})^{-1})(1 - (2\sigma_q\sqrt{\pi})^{-1})},
\end{align}
where $\lambda(.)$ in (\ref{closedform}) is defined as follows:
\begin{align}
\label{lambdaexpression}
\lambda(\Theta_p,\Theta_q) = \exp\left( -\frac{(\mu_p - \mu_q)^2}{2(\sigma^2_p +\sigma^2_q)} \right)\left(\pi(\sigma^2_p + \sigma^2_q)\right)^{-\frac{1}{2}}.
\end{align}
Hence, the conditions in (\ref{connectionformalised}) materialize when
\begin{align}
&||p||_{\mathcal{L}^{2}} < 1 \Longleftrightarrow \sigma_p > \frac{1}{2\sqrt{\pi}} \hspace{0.15in} \emph{and} \hspace{0.15in} ||q||_{\mathcal{L}^{2}} < 1 \Longleftrightarrow \sigma_q > \frac{1}{2\sqrt{\pi}} \label{conditionssigmatwo}.
\end{align}
\end{prop}
\begin{proof}
The proof follows from products of Gaussian densities on $\R$ as they appear in (\ref{measuretheorticexp}). More specifically, we have the following product:
\begin{align}
p(x)q(x) &= \frac{1}{2\pi \sigma_p\sigma_q}\exp\left( - \left(\frac{(x - \mu_p)^2}{2\sigma^2_p} + \frac{(x - \mu_q)^2}{2\sigma^2_q} \right)\right) = \frac{h_{pq}}{\sqrt{2\pi} \sigma_{pq}}\exp\left( - \frac{(x - \mu_{pq})^2}{2\sigma^2_{pq}}\right) \label{productexpression},
\end{align}
for any $x\in\R$, where the new terms in (\ref{productexpression}) are defined as follows:
\begin{align}
\mu_{pq} = \frac{\mu_p\sigma^2_q + \mu_q\sigma^2_p}{\sigma^2_p + \sigma^2_q}, \,\,\,\,\, \sigma_{pq} = \sqrt{\frac{\sigma^2_p\sigma^2_q}{\sigma^2_p + \sigma^2_q}} , \,\,\,\,\, h_{pq} = \frac{1}{\sqrt{2\pi(\sigma^2_p + \sigma^2_q)}}\exp\left( - \frac{(\mu_p - \mu_{q})^2}{2(\sigma^2_p + \sigma^2_q)}\right). \label{newterms}
\end{align}
Therefore, we have the following:
\begin{align}
4\int_{\Omega}\dd\P(\omega)\dd\Q(\omega) = 4 \frac{h_{pq}}{\sqrt{2\pi} \sigma_{pq}}\int_{\R} \exp\left( - \frac{(x - \mu_{pq})^2}{2\sigma^2_{pq}}\right)\dd x = 4h_{pq} = 2\sqrt{2}\lambda(\Theta_p,\Theta_q) \notag
\end{align}
where $\lambda(.)$ is as given in (\ref{lambdaexpression}). Accordingly, we also have the following products integrals:
\begin{align}
\int_{\Omega}\dd\P(\omega)^2 &= \frac{h_{pp}}{\sqrt{2\pi} \sigma_{pp}}\int_{\R} \exp\left( - \frac{(x - \mu_{pp})^2}{2\sigma^2_{pp}}\right)\dd x = h_{pp} = \frac{1}{2\sigma_{p}\sqrt{\pi}} \notag \\
\int_{\Omega}\dd\Q(\omega)^2 &= \frac{h_{qq}}{\sqrt{2\pi} \sigma_{qq}}\int_{\R} \exp\left( - \frac{(x - \mu_{qq})^2}{2\sigma^2_{qq}}\right)\dd x = h_{qq} = \frac{1}{2\sigma_{q}\sqrt{\pi}}, \notag
\end{align}
and the expression given in (\ref{closedform}) follows. Finally, the statement in (\ref{conditionssigmatwo}) follows from the denominator of (\ref{closedform}) for $||p||_{\mathcal{L}^{2}} < 1$ and $||q||_{\mathcal{L}^{2}}< 1$. 
\end{proof}
Using the closed-form expression (\ref{closedform}), we can study the asymptotic behaviour of $\Psi(\P,\Q)$ across relevant parametric limits, for which we omit the proof since it follows directly from (\ref{closedform}).
\begin{coro}
The following limits hold:
\begin{align}
\Psi(\P,\Q) \rightarrow \frac{2\sigma_p\sqrt{\pi}}{2\sigma^2_p\pi - \sigma_p\sqrt{\pi}} \hspace{0.1in} \text{as $\sigma_q \rightarrow \infty$} \,\,\,\, \text{and} \,\,\,\,
\Psi(\P,\Q) \rightarrow \frac{2\sigma_q\sqrt{\pi}}{2\sigma^2_q\pi - \sigma_q\sqrt{\pi}} \hspace{0.1in} \text{as $\sigma_p \rightarrow \infty$} \label{limitone}
\end{align}
for any $\mu_p,\mu_q$. In addition, the double-limit satisfies
\begin{align}
\Psi(\P,\Q) &\rightarrow 0 \hspace{0.1in} \text{as $\sigma_p, \sigma_q \rightarrow \infty$} \label{limitthree}
\end{align}
for any $\mu_p,\mu_q$. 
\end{coro}
It can be seen from (\ref{limitone})-(\ref{limitthree}) that the mean parameters $\mu_p$ and $\mu_q$, when fixed, lose their impact on $\Psi$ asymptotically as one or both of the variance parameters diverge -- i.e. if any of the Gaussian measures have a significantly large variance, then the (fixed) mean parameter is decreasingly important in distinguishing the difference between those measures. From (\ref{limitthree}) we further see that two Gaussian measures with significantly large variances are less distinguishable from each other. 

\subsection*{\Poincare Disc Radius Generalisation}
Since the core objective of this paper has been to present the aforementioned geometric connection, our focus was on the \Poincare \emph{unit} disc. However, conditions in (\ref{conditionssigmatwo}) may render (\ref{closedform}) too restrictive for practical purposes. Accordingly, one can generalise the disc (\ref{hyoerbolicspacetwo}) as follows:
\begin{align}
\label{hyoerbolicspacetwogeneral}
\mathcal{D}=\{z^{*}\in\mathbb{C} : ||z^{*}|| < R \},
\end{align}
for any radius $0 < R < \infty$. Accordingly, the curvature of the disc is $K= - 1/ R^2$. Following the same logic, we reach a generalised $\mathcal{L}^{2}$-embedded $\Psi$ as 
\begin{align} 
\Psi(\P,\Q)&=\frac{2R^2||p-q ||_{\mathcal{L}^{2}}^{2}}{(R^2-||p||_{\mathcal{L}^{2}}^{2})(R^2-||q||_{\mathcal{L}^{2}}^{2})} =\frac{2R^2\int_{\Omega}\left( \dd\P(\omega) - \dd\Q(\omega) \right)^2}{\left(R^2 - \int_{\Omega}\dd\P(\omega)^2\right)\left(R^2 - \int_{\Omega}\dd\Q(\omega)^2\right)}, \label{measuretheorticexpgeneral}
\end{align}
given that the following conditions are satisfied: 
\begin{align}
||p||_{\mathcal{L}^{2}} < R \,\,\, \text{and} \,\,\, ||q||_{\mathcal{L}^{2}}< R \label{connectionformalisedgeneral}. 
\end{align}
Analogous to the unit case, (\ref{measuretheorticexpgeneral}) acts as an isometric invariant over the infinite-dimensional real Hilbert ball of radius $R$, defined by $\mathbb{B}_R(\mathcal{L}^2) = \{ f \in \mathcal{L}^2 : \|f\|_{\mathcal{L}^2} < R \}$, equipped with its scaled Poincar\'e-type hyperbolic metric. Accordingly, we have the following generalisation for Proposition \ref{mainpropgeometric}, for which we omit the proof in order to save space.
\begin{prop}
\label{mainpropgeometricgeneral}
The divergence (\ref{measuretheorticexpgeneral}) on $\R$ is given by
\begin{align} 
\label{closedformgeneral}
\Psi(\P,\Q)&= \frac{R^2\left( (\sigma_p\sqrt{\pi})^{-1} - 2\sqrt{2}\lambda(\Theta_p,\Theta_q) + (\sigma_q\sqrt{\pi})^{-1}\right)}{(R^2 - (2\sigma_p\sqrt{\pi})^{-1})(R^2 - (2\sigma_q\sqrt{\pi})^{-1})},
\end{align}
where $\lambda(.)$ in (\ref{closedformgeneral}) is defined as follows:
\begin{align}
\label{lambdaexpressiongeneral}
\lambda(\Theta_p,\Theta_q) = \exp\left( -\frac{(\mu_p - \mu_q)^2}{2(\sigma^2_p +\sigma^2_q)} \right)\left(\pi(\sigma^2_p + \sigma^2_q)\right)^{-\frac{1}{2}}.
\end{align}
Hence, the conditions in (\ref{connectionformalisedgeneral}) materialize when
\begin{align}
&||p||_{\mathcal{L}^{2}} < R \Longleftrightarrow \sigma_p > \frac{1}{2R^2\sqrt{\pi}} \hspace{0.15in} \emph{and} \hspace{0.15in} ||q||_{\mathcal{L}^{2}} < R \Longleftrightarrow \sigma_q > \frac{1}{2R^2\sqrt{\pi}} \label{conditionssigmatwogeneral}.
\end{align}
\end{prop}
Proposition \ref{mainpropgeometricgeneral} introduces $0 < R < \infty$ that can be controlled based on the application and data at hand. On the other hand, rewriting $\Psi$, we highlight that 
\begin{align}
\Psi(\P,\Q) = \frac{2||p-q ||_{\mathcal{L}^{2}}^{2}}{R^2\left(1-\frac{||p||_{\mathcal{L}^{2}}^{2}}{R^2}\right)\left(1-\frac{||q||_{\mathcal{L}^{2}}^{2}}{R^2}\right)} \rightarrow 0 \,\,\,\, \text{as} \,\,\,\, R \rightarrow \infty. \label{convergencezero} 
\end{align}
Accordingly, the divergence metric gets progressively weaker to distinguish $\P$ and $\Q$ as $R \rightarrow \infty$, where the curvature $K \rightarrow 0$ in the limit, and the geometry converges to that of (flat) Euclidean. 
\begin{rem}
The radius $R$ can be interpreted as a curvature tuner between the \Poincare disc and the Euclidean space. The choice of curvature provides a different $\Psi$ value between $\P$ and $\Q$, which can be employed as a learned hyperparameter.
\end{rem}
As an example, the closed-form expression (\ref{closedformgeneral}) can be used to compare any two $\R$-valued Gaussian processes on a distributional basis subject to condition (\ref{conditionssigmatwogeneral}). As a canonical family, we can analyze two $\R$-valued Brownian motions, as per below.
\begin{ex}
\label{examplegaussianprocesscomparegeneral}
Let $\{W^{(p)}_t\}_{t\geq0}$ be a drifted $\P$-Brownian motion with 
\begin{align}
\E^{\P}[W^{(p)}_t] = t\mu_p \,\,\,\,\, \text{and} \,\,\,\,\, \text{Var}^{\P}[W^{(p)}_t] = t\sigma^2_p. \notag
\end{align}
Also, let $\{W^{(q)}_t\}_{t\geq0}$ be a drifted $\Q$-Brownian motion with 
\begin{align}
\E^{\Q}[W^{(q)}_t] = t\mu_q\,\,\,\,\, \text{and} \,\,\,\,\, \text{Var}^{\Q}[W^{(q)}_t] = t\sigma^2_q. \notag 
\end{align}
Define two Gaussian processes $\{Z^{(p)}_t\}_{t\geq0}$ and $\{Z^{(q)}_t\}_{t\geq0}$ by
\begin{align}
Z^{(p)}_t = Z^{(p)} + W^{(p)}_t \,\,\,\, \text{and} \,\,\,\, Z^{(q)}_t = Z^{(q)} + W^{(q)}_t \notag
\end{align}
for every $t\geq 0$, where $Z^{(p)}$ is mutually independent Gaussian with 
\begin{align}
\E^{\P}[Z^{(p)}] = \eta_p \,\,\,\, \text{and} \,\,\,\, \text{Var}^{\P}[Z^{(p)}] = \varsigma^2_p > (4R^4\pi)^{-1}, \notag 
\end{align}
and $Z^{(q)}$ is mutually independent Gaussian with 
\begin{align}
\E^{\Q}[Z^{(q)}] = \eta_q \,\,\,\, \text{and} \,\,\,\, \text{Var}^{\Q}[Z^{(q)}] = \varsigma^2_q > (4R^4\pi)^{-1}. \notag 
\end{align}
Then, the $\Psi$-divergence between $\{Z^{(p)}_t\}_{t\geq0}$ and $\{Z^{(q)}_t\}_{t\geq0}$ is as follows:
\begin{align} 
\label{brownianexamplegeneralextended}
\Psi(\P,\Q)&= \frac{R^2\left(\left(\sqrt{(\varsigma^2_p + t\sigma^2_p)\pi}\right)^{-1} - 2\sqrt{2}\lambda_{t}(\Theta_p,\Theta_q) + \left(\sqrt{(\varsigma^2_q + t\sigma^2_q)\pi}\right)^{-1}\right)}{\left(R^2 - \left(2\sqrt{\left(\varsigma^2_p + t\sigma^2_p\right)\pi}\right)^{-1}\right)\left(R^2 - \left(2\sqrt{\left(\varsigma^2_q + t\sigma^2_q\right)\pi}\right)^{-1}\right)},
\end{align}
for every $t\geq 0$, given that
\begin{align}
\lambda_t(\Theta_p,\Theta_q) = \exp\left( -\frac{(\eta_p - \eta_q + t(\mu_p  - \mu_q) )^2}{2(\varsigma^2_p +\varsigma^2_q + t(\sigma^2_p  + \sigma^2_q))} \right)\left(\pi(\varsigma^2_p +\varsigma^2_q + t(\sigma^2_p  + \sigma^2_q))\right)^{-\frac{1}{2}}. \notag
\end{align}
Note that $\Psi$-divergence in (\ref{brownianexamplegeneralextended}) holds for every $t\geq 0$, since 
\begin{align}
\varsigma_p > (2R^2\sqrt{\pi})^{-1} \,\,\,\, \text{and} \,\,\,\, \varsigma_q> (2R^2\sqrt{\pi})^{-1} \notag 
\end{align}
for $Z^{(p)}$ and $Z^{(q)}$ as Gaussian random initial values, mutually independent from $\{W^{(p)}_t\}_{t\geq0}$ and $\{W^{(q)}_t\}_{t\geq0}$, respectively. Note that if we would have $\{Z^{(p)}_t\}_{t\geq0} = \{W^{(p)}_t\}_{t\geq0}$ and $\{Z^{(q)}_t\}_{t\geq0} = \{W^{(q)}_t\}_{t\geq0}$ instead, then
\begin{align} 
\label{brownianexample}
\Psi(\P,\Q)&= \frac{ R^2\left((\sigma_p\sqrt{t\pi})^{-1} - 2\sqrt{2}\lambda_{t}(\Theta_p,\Theta_q) + (\sigma_q\sqrt{t\pi})^{-1}\right)}{(R^2 - (2\sigma_p\sqrt{t\pi})^{-1})(R^2 - (2\sigma_q\sqrt{t\pi})^{-1})} \,\,\,\, \text{\emph{for} $t > \max\left(\frac{1}{4R^4\sigma^2_p\pi} \, , \, \frac{1}{4R^4\sigma^2_q\pi}\right)$},
\end{align}
given that
\begin{align}
\lambda_t(\Theta_p,\Theta_q) = \exp\left( -\frac{t(\mu_p - \mu_q)^2}{2(\sigma^2_p +\sigma^2_q)} \right)\left(t\pi(\sigma^2_p + \sigma^2_q)\right)^{-\frac{1}{2}}. \notag
\end{align}
\end{ex}
For certain practical uses, if (\ref{convergencezero}) is a limitation, one can produce a scaled $\Psi$ measure $\hat{\Psi}$ given by the following:
\begin{align}
\hat{\Psi}(\P,\Q) = \frac{R^2}{2}\Psi(\P,\Q), \notag
\end{align}
which satisfies
\begin{align}
\hat{\Psi}(\P,\Q) \rightarrow ||p-q ||_{\mathcal{L}^{2}}^{2} = \int_{\Omega}\left( \dd\P(\omega) - \dd\Q(\omega) \right)^2 \,\,\,\, \text{as} \,\,\,\, R \rightarrow \infty. \label{eucledianconvergence}
\end{align}
Although the scaled divergence measure $\hat{\Psi}$ loses its direct geometric interpretation as per above, it converges to the squared Euclidean distance as $R$ goes to infinity. It is worthwhile to note that property (\ref{eucledianconvergence}) may be more suitable than (\ref{convergencezero}) for certain applications.

\subsection*{Multivariate Generalisation}
The closed-form expression in (\ref{closedform}) can be generalised to multivariate Gaussian measures. Accordingly, let $\P$ and $\Q$ be Gaussian measures on $\R^n$, where $p:\R^n\rightarrow\R_+$ and $q:\R^n\rightarrow\R_+$ are the corresponding Gaussian probability density functions for some $n\geq 1$. Denote the parameters of $p$ and $q$ as $\boldsymbol{\Theta}_p = \{\boldsymbol{\mu}_p,\boldsymbol{\Sigma}_p\}$ and $\boldsymbol{\Theta}_q = \{\boldsymbol{\mu}_q,\boldsymbol{\Sigma}_q\}$, respectively. Hence, writing $|\cdot|$ as the determinant, we have the following density expressions:
\begin{align}
p(\boldsymbol{x}) &= \frac{1}{(2\pi)^{\frac{n}{2}}\sqrt{|\boldsymbol{\Sigma}_p|}}\exp\left(-\frac{1}{2}(\boldsymbol{x} - \boldsymbol{\mu}_p )^{\tp}\boldsymbol{\Sigma}_p^{-1}(\boldsymbol{x} - \boldsymbol{\mu}_p )\right) = \exp\left( \boldsymbol{\Lambda}_p + \hat{\boldsymbol{\mu}}^{\tp}_p\boldsymbol{x} - \frac{1}{2}\boldsymbol{x}^{\tp}\boldsymbol{\Sigma}_p^{-1}\boldsymbol{x} \right) \notag
\end{align}
and
\begin{align}
q(\boldsymbol{x}) &= \frac{1}{(2\pi)^{\frac{n}{2}}\sqrt{|\boldsymbol{\Sigma}_q|}}\exp\left(-\frac{1}{2}(\boldsymbol{x} - \boldsymbol{\mu}_q )^{\tp}\boldsymbol{\Sigma}_q^{-1}(\boldsymbol{x} - \boldsymbol{\mu}_q )\right) = \exp\left( \boldsymbol{\Lambda}_q + \hat{\boldsymbol{\mu}}^{\tp}_q\boldsymbol{x} - \frac{1}{2}\boldsymbol{x}^{\tp}\boldsymbol{\Sigma}_q^{-1}\boldsymbol{x} \right), \notag
\end{align}
given that
\begin{align}
&\hat{\boldsymbol{\mu}}_i = \boldsymbol{\Sigma}^{-1}_i\boldsymbol{\mu}_i \notag \\
&\boldsymbol{\Lambda}_i = -\frac{1}{2}\left(n \log(2\pi) - \log(|\boldsymbol{\Sigma}^{-1}_i|) + \hat{\boldsymbol{\mu}}^{\tp}_i\boldsymbol{\Sigma}_i\hat{\boldsymbol{\mu}}_i \right), \notag
\end{align}
for $i\in\{p,q\}$. Therefore, we can write the following:
\begin{align}
p(\boldsymbol{x})q(\boldsymbol{x}) &= \exp\left(\boldsymbol{\Lambda}_p + \boldsymbol{\Lambda}_q + \left(\hat{\boldsymbol{\mu}}^{\tp}_p + \hat{\boldsymbol{\mu}}^{\tp}_q \right)\boldsymbol{x} - \frac{1}{2}\boldsymbol{x}^{\tp}\left(\boldsymbol{\Sigma}_p^{-1} + \boldsymbol{\Sigma}_q^{-1}\right)\boldsymbol{x}\right) \notag \\
&= \exp\left(\boldsymbol{\Lambda}_p + \boldsymbol{\Lambda}_q + \hat{\boldsymbol{\mu}}^{\tp}_{pq}\boldsymbol{x} - \frac{1}{2}\boldsymbol{x}^{\tp}\boldsymbol{\Sigma}_{pq}^{-1}\boldsymbol{x}\right) \notag \\
&= \exp\left(\boldsymbol{\Lambda}_p + \boldsymbol{\Lambda}_q - \boldsymbol{\Lambda}_{pq}\right) \exp\left(\boldsymbol{\Lambda}_{pq} + \hat{\boldsymbol{\mu}}^{\tp}_{pq}\boldsymbol{x} - \frac{1}{2}\boldsymbol{x}^{\tp}\boldsymbol{\Sigma}_{pq}^{-1}\boldsymbol{x}\right) \notag
\end{align}
where we have
\begin{align}
\boldsymbol{\Lambda}_{pq} = -\frac{1}{2}\left(n \log(2\pi) - \log(|\boldsymbol{\Sigma}^{-1}_{pq}|) + \hat{\boldsymbol{\mu}}^{\tp}_{pq}\boldsymbol{\Sigma}_{pq}\hat{\boldsymbol{\mu}}_{pq} \right).\notag
\end{align}
It follows that
\begin{align}
\int_{\Omega}\dd\P(\omega)\dd\Q(\omega) = \exp\left(\boldsymbol{\Lambda}_p + \boldsymbol{\Lambda}_q - \boldsymbol{\Lambda}_{pq}\right),
\end{align}
and
\begin{align}
\int_{\Omega}\dd\P(\omega)^2 = \exp\left(2\boldsymbol{\Lambda}_p - \boldsymbol{\Lambda}_{pp}\right), \,\,\,\, \int_{\Omega}\dd\Q(\omega)^2 = \exp\left(2\boldsymbol{\Lambda}_q - \boldsymbol{\Lambda}_{qq}\right). \notag
\end{align}
Finally, as we insert these terms into (\ref{measuretheorticexp}), we get the following expession:
\begin{align} 
\label{closedformgeneralised}
\Psi(\P,\Q)&= \frac{2\left(\exp\left(2\boldsymbol{\Lambda}_p - \boldsymbol{\Lambda}_{pp}\right) - 2\exp\left(\boldsymbol{\Lambda}_p + \boldsymbol{\Lambda}_q - \boldsymbol{\Lambda}_{pq}\right) + \exp\left(2\boldsymbol{\Lambda}_q - \boldsymbol{\Lambda}_{qq}\right) \right)}{\left(1 - \exp\left(2\boldsymbol{\Lambda}_p - \boldsymbol{\Lambda}_{pp}\right)  \right) \left(1 - \exp\left(2\boldsymbol{\Lambda}_q - \boldsymbol{\Lambda}_{qq}\right) \right)},
\end{align}
as a multivariate generalisation to (\ref{closedform}) in Proposition \ref{mainpropgeometric}, with the boundary conditions:
\begin{align}
|\mathbf{\Sigma}_p| > \frac{1}{(4\pi)^n} \quad \text{and} \quad |\mathbf{\Sigma}_q| > \frac{1}{(4\pi)^n}
\end{align}
given that we have the following:
\begin{align}
||p||_{\mathcal{L}^{2}}^{2} = \left(2^n \pi^{n/2} \sqrt{|\mathbf{\Sigma}_p|}\right)^{-1} \,\,\,\, \text{and} \,\,\,\, ||q||_{\mathcal{L}^{2}}^{2} = \left(2^n \pi^{n/2} \sqrt{|\mathbf{\Sigma}_q|}\right)^{-1}. \notag
\end{align}
The expression in (\ref{closedformgeneralised}) can be further simplified, which we leave to the interested reader. Note that under the \Poincare disc radius generalisation, we have
\begin{align} 
\label{closedformgeneralisedgeneral}
\Psi(\P,\Q)&= \frac{2R^2\left(\exp\left(2\boldsymbol{\Lambda}_p - \boldsymbol{\Lambda}_{pp}\right) - 2\exp\left(\boldsymbol{\Lambda}_p + \boldsymbol{\Lambda}_q - \boldsymbol{\Lambda}_{pq}\right) + \exp\left(2\boldsymbol{\Lambda}_q - \boldsymbol{\Lambda}_{qq}\right) \right)}{\left(R^2 - \exp\left(2\boldsymbol{\Lambda}_p - \boldsymbol{\Lambda}_{pp}\right)  \right) \left(R^2 - \exp\left(2\boldsymbol{\Lambda}_q - \boldsymbol{\Lambda}_{qq}\right) \right)},
\end{align}
as an extension to (\ref{closedformgeneralised}), where we have the boundary conditions:
\begin{align}
|\mathbf{\Sigma}_p| > \frac{1}{(4\pi)^nR^4} \quad \text{and} \quad |\mathbf{\Sigma}_q| > \frac{1}{(4\pi)^nR^4}, \label{lowerboundaryhghdim}
\end{align}
which accounts for the choice of $R$ as the curvature hyperparameter. We shall highlight that the lower bound on the generalized variance decays exponentially as $n$ increases. Accordingly, in high-dimensional spaces, the statistical requirement to prevent the proposed model (\ref{closedformgeneralisedgeneral}) from breaking down becomes progressively small. Thus, under the generalised framework, high-dimensional models are allowed to hold increasingly sharp, low generalised  variance representations without breaching the boundary dictated by the \Poincare disc.
\begin{rem}
In statistical machine learning, high-dimensional spaces usually become challenging in practice, commonly referred to as the curse of dimensionality. In (\ref{lowerboundaryhghdim}), high dimensionality acts as an increasingly strong buffer to relax the lower boundary constraint on the generalised variance.
\end{rem}

\section{Conclusion}
Given the hyperbolic geometric connection of (\ref{eq:comparisonpoint1}) to the spherical squared-Hellinger distance, we are encouraged to communicate the $\mathcal{L}^{2}$-embedded (\ref{closedform}) -- and more generally (\ref{closedformgeneralisedgeneral}) -- as an alternative measure of statistical divergence between Gaussian distributions in agreement with our Proposal \ref{dualitystatementinitalpresentation}. We refer to \cite{1c,1b} for applications of divergence measures in variational Bayesian inference. We also refer to \cite{3} for quantum counterparts of the squared-Hellinger distance within quantum information theory, which lend themselves for future research in connection with the spherical-hyperbolic duality we presented. We hope to apply our proposed measure in various practical fields that may benefit from our proposal, especially where limitations of applying flat Euclidean metrics on inherently curved data structures present a fundamental challenge in the machine learning task at hand.

\subsection*{Appendix}
In this section, we shall provide visual heatmaps of the proposed divergence metric provided in (\ref{closedformgeneralisedgeneral}) for demonstration purposes. In order to save space, we shall fix the curvature tuner $R=1$, and work with (\ref{closedform}) on the \Poincare unit disc. We highlight that the numerical analysis provided below is not exhaustive, and one may visualize different patterns of the proposed divergence measure.
First, we shall provide examples as we set $\sigma_p$ and $\sigma_q$ as the axes of the heatmaps, and change $\mu_p$ and $\mu_q$.
\\
\begin{figure}[htbp] 
  \centering
  \hfill
  \includegraphics[width=0.475\linewidth]{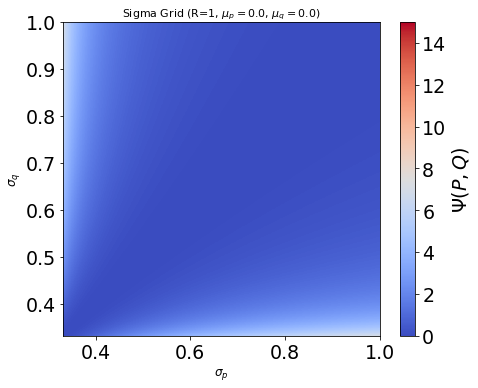}%
  \hfill\hfill
  \includegraphics[width=0.475\linewidth]{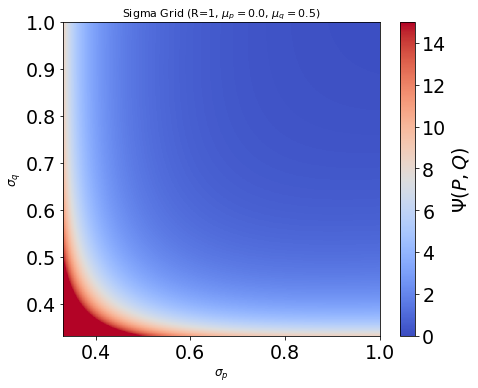}%
  \hfill\mbox{}
  
  \vspace{1em} 
  
  \hfill
  \includegraphics[width=0.475\linewidth]{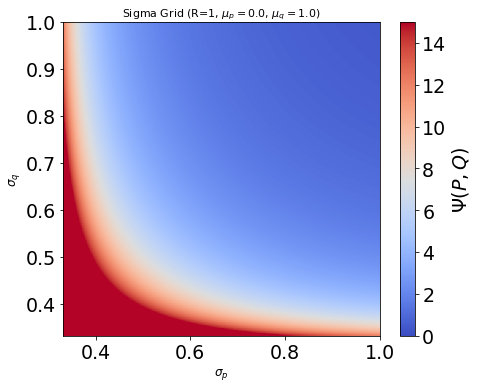}%
  \hfill\hfill
  \includegraphics[width=0.475\linewidth]{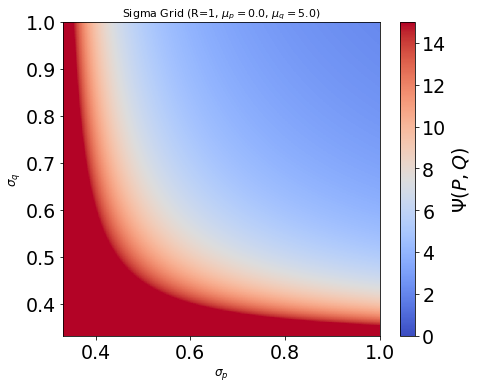}%
  \hfill\mbox{}
  
  \caption[Metric]{Heatmaps across different $\mu_p$ vs. $\mu_q$ input pairs. $R = 1$.}
  \label{fig:sigma_heatmaps}
\end{figure}
\\
The area of higher values (shown in red) of the divergence measure $\Psi(\P,\Q)$ increases as the asymmetry of the parameters $\mu_p$ and $\mu_q$ increases in magnitude, as per above. Second, we shall provide examples as we set $\mu_p$ and $\mu_q$ as the axes of the heatmaps, and change $\sigma_p$ and $\sigma_q$ pairs.
\\
\begin{figure}[htbp] 
  \centering
  \hfill
  \includegraphics[width=0.475\linewidth]{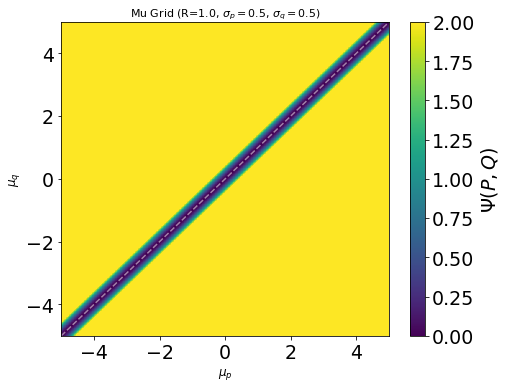}%
  \hfill\hfill
  \includegraphics[width=0.475\linewidth]{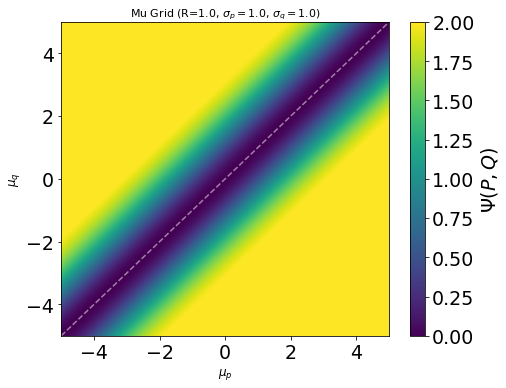}%
  \hfill\mbox{}
  
  \vspace{1em} 
  
  \hfill
  \includegraphics[width=0.475\linewidth]{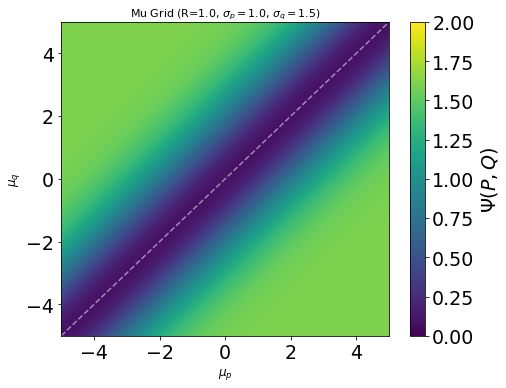}%
  \hfill\hfill
  \includegraphics[width=0.475\linewidth]{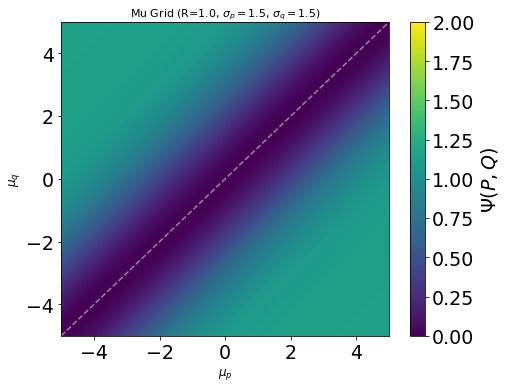}%
  \hfill\mbox{}
  
  \caption[Metric]{Heatmaps across different $\sigma_p$ vs. $\sigma_q$ input pairs. $R = 1$.}
  \label{fig:mu_heatmaps}
\end{figure}
\\
The area of lower values (shown in blue) of the divergence measure $\Psi(\P,\Q)$ increases as the parameters $\sigma_p$ and $\sigma_q$ increase in magnitude, as per above.

\end{document}